\documentclass{article}

\usepackage{amsmath}
\usepackage{amssymb}
\usepackage{amsthm}
\usepackage{url}
\usepackage{stmaryrd}

\newcounter{global}
\theoremstyle{definition}
\newtheorem{definition}[global]{Definition}
\theoremstyle{plain}
\newtheorem{theorem}[global]{Theorem}
\newtheorem{lemma}[global]{Lemma}

\newtheoremstyle{note}{}{}{}{}{\itshape}{.}{.5em}{}
\theoremstyle{note}
\newtheorem{remark}{Remark}%
\newtheorem{example}{Example}%
\renewcommand\section{%
  \@startsection {section}{1}{\z@}%
  {-3.5ex \@plus -1ex \@minus -.2ex}%
  {2.3ex \@plus.2ex}%
  {\normalfont\large\bfseries}}
\pdfmapline{=eurmo10 EURM10 ".167 SlantFont" <eurm10.pfb}
\pdfmapline{=eurmo10 EURM10 " .167 SlantFont" <eurm10.pfb}

\DeclareFontFamily{U}{eurmo}{}
\DeclareFontShape{U}{eurmo}{m}{n}{<-6>eurmo10<6-8>eurmo10<8->eurmo10}{}
\DeclareMathAlphabet{\eurmo}{U}{eurmo}{m}{n}

\def\itm#1{{\rm(\textit{\romannumeral#1})}}

\def\e#1{\eurmo{#1}}

\def\ineq{\ensuremath{\preccurlyeq}}
\def\tu#1{\left<#1\right>}
\def\br#1{\left[#1\right]}
\newcommand{\st}[1][M]{{\ensuremath{\mathbf{#1}}}}
\newcommand{\alg}[1][M]{\ensuremath{\tu{#1, F^{\st[#1]}}}}
\newcommand{\loalg}[1][M]{\ensuremath{\tu{#1, \ineq^{\st[#1]}, F^{\st[#1]}}}}
\let\plL=\L
\renewcommand{\L}{{\st[L]}}

\newcommand{\TX}[1][X]{{{\st[T]_{\!F}(#1)}}}

\newcommand{\qc}[1][\xi]{\ensuremath{#1}}
\newcommand{\cl}[2][\xi]{{\ensuremath{\br{#2}_{#1}}}}
\newcommand{\val}[2][\st,v]{%
  \ensuremath{\left\lVert%
      \bgroup #2\egroup%
    \right\rVert_{#1}}}
\newcommand{\prv}[2][\Sigma]{%
  \ensuremath{\left\lvert%
      \bgroup #2\egroup%
    \right\rvert_{#1}}}
\newcommand{\home}[1]{\ensuremath{{#1^{\scriptscriptstyle\sharp}}}}
\newcommand{\Mod}[1][\Sigma]{{\mathop{\mathrm{Mod\!}}{(#1)}}}
\let\oldphi=\phi
\let\phi=\varphi
\let\varphi=\oldphi
\newcommand{\Fml}{\mathit{Fml}}
\newcommand{\semcl}[1][\Sigma]{#1^{\vDash}}
\newcommand{\syncl}[1][\Sigma]{#1^{\vdash}}
\newcommand{\cdotsym}{\,\cdot\,}
\def\logand{\mathop{\binampersand}}

\begin{document}

\title{Fuzzy inequational logic}

\date{\normalsize%
  Dept. Computer Science, Palacky University, Olomouc}

\author{Vilem Vychodil\footnote{%
    e-mail: \texttt{vychodil@binghamton.edu},
    phone: +420 585 634 705,
    fax: +420 585 411 643}}

\maketitle

\begin{abstract}
  We present a logic for reasoning about graded inequalities which
  generalizes the ordinary inequational logic used in universal algebra.
  The logic deals with atomic predicate formulas of the form of
  inequalities between terms and formalizes their semantic entailment
  and provability in graded setting which allows to draw partially true
  conclusions from partially true assumptions. We follow the Pavelka
  approach and define general degrees of semantic entailment and provability
  using complete residuated lattices as structures of truth degrees. We
  prove the logic is Pavelka-style complete. Furthermore, we present
  a logic for reasoning about graded if-then rules which is obtained
  as particular case of the general result.
\end{abstract}

\section{Introduction}\label{sec:intro}
In this paper, we introduce a general logic for approximate reasoning about
atomic predicate formulas which take form of inequalities between terms. Such
formulas, called inequalities, are essential in the classic theory of varieties
of ordered algebras~\cite{Bloom} since the varieties are
exactly the classes of ordered algebras which are definable by sets of
inequalities. We extend the classic logic for reasoning about inequalities
by considering \emph{degrees} to which one considers the inequalities valid.
We would like to stress that our approach is truth-functional and
the degrees we use are interpreted as the \emph{degrees of truth} and they
should not be confused or interchanged with degrees that appear in other
formalisms and uncertainty theories (like the degrees of \emph{belief}).
We assume that the degrees come from general structures of truth degrees.
In particular, we use complete residuated
lattices~\cite{Bel:FRS,GaJiKoOn:RL,Gog:Lic}.

In the proposed logic, we introduce two types of entailment: First,
a \emph{semantic entailment} which is based on evaluating inequalities in
particular fuzzy structures called algebras with fuzzy orders~\cite{Vy:Awfo}.
Using algebras with fuzzy orders as models, we are able to introduce degrees
to which inequalities semantically follow from collections of partially valid
inequalities. Second, we introduce a graded notion of \emph{provability}
(\emph{syntactic entailment}) which allows us to infer partially valid
conclusions from collections of partially valid inequalities. The notion of
graded provability is defined using a specific deductive system which consists
of axioms and three deduction rules. We prove that our logic is complete
in that the degrees of semantic entailment coincide with the degrees
of provability. This type of graded completeness is called Pavelka-style
completeness~\cite{Haj:MFL} after J.~Pavelka who, inspired by the
influential paper by J.\,A.~Goguen~\cite{Gog:Lic}, presented the
general concept in~\cite{Pav:Ofl1,Pav:Ofl2,Pav:Ofl3} and studied
Pavelka-style complete propositional logics. A thorough and general
treatment of logics with this style of completeness is presented
in~\cite{Ger:FL}.

We consider the completeness result to be the main result of this paper.
In addition to that, we present an application of the result showing
a complete axiomatization of a logic for reasoning about graded if-then
rules called attribute implications. Such rules, sometimes used under different
names, are formulas which play important roles in several disciplines concerned
with data analysis and management such as the formal concept analysis~\cite{GaWi:FCA}
and relational databases~\cite{Mai:TRD}. We show in the paper that the rules
can be treated as particular inequalities and that using our general result
we may obtain a complete logic for approximate reasoning with such inequalities.
By making this observation, we contribute to the area of reasoning with graded
if-then rules and present an alternative to the approaches
in~\cite{BeVy:ADfDwG,Po:FB} which may further be explored.

Previous results which are related to our paper include the fuzzy
equational logic~\cite{Be:Fel} which introduced Pavelka-style logic for
reasoning about graded equalities and fuzzy Horn logic dealing with
implications between graded equalities~\cite{BeVy:FHLI}. Our logic can be
seen as a generalization of the fuzzy equational logic.
Indeed, from the syntactic point of view, it is a logic which results from
the fuzzy equational logic by omitting the deduction rule of symmetry.
From the semantic point of view, the present logic uses more general
models---algebras with fuzzy orders~\cite{Vy:Awfo} instead of algebras
with fuzzy equalities~\cite{BeVy:Awfs}. A survey of results on fuzzy
equational logic can be found in~\cite{BeVy:FEL}.

This paper is organized as follows. In Section~\ref{sec:prelim},
we present preliminaries from residuated structures of truth degrees and
algebras with fuzzy orders. In Section~\ref{sec:synsem}, we introduce our
logic and present the central notions of semantic and syntactic entailments.
In Section~\ref{sec:compl}, we show that our logic is syntactico-semantically
complete in Pavelka style. In Section~\ref{sec:fal}, we present an application
of the general completeness result provided in Section~\ref{sec:compl} by
showing a general logic of attribute implications with
a complete Pavelka-style axiomatization.

\section{Preliminaries}\label{sec:prelim}
In this section, we present basic notions of complete residuated lattices
which appear in the fuzzy inequational logic as the structures of truth degrees.
Moreover, we present algebras with fuzzy orders which are used
as the basic semantic structures in the fuzzy inequational logic.

\subsection{Complete Residuated Lattices}
A complete (integral commutative)
residuated lattice \cite{Bel:FRS,GaJiKoOn:RL}
is an algebra $\L=\langle L,\wedge,\vee,\otimes,\rightarrow,0,1\rangle$
where
$\langle L,\wedge,\vee,0,1 \rangle$ is a complete lattice,
$\langle L,\otimes,1 \rangle$ is a commutative monoid, and
$\otimes$ and $\rightarrow$ satisfy the adjointness property:
$a \otimes b \leq c$ if{}f $a \leq b \rightarrow c$
($a,b,c \in L$). Examples of complete residuated
lattices include structures on the real unit interval given
by left-continuous t-norms~\cite{EsGo:MTL,Haj:MFL,KMP:TN}
as well as finite structures of degrees.

Given $\L$ and $M \ne \emptyset$, an $\L$-set $A$ in $M$
(or a fuzzy set in $M$ using degrees in $L$) is a map $A\!: M \to L$.
For $\e{a} \in M$, the degree $A(\e{a}) \in L$ is interpreted as the
degree to which $\e{a}$ belongs to $A$.
Analogously, a binary $\L$-relation $R$
on $M$ is a map $R\!: M \times M \to L$. For $\e{a},\e{b} \in M$,
the degree $R(\e{a},\e{b}) \in L$ is interpreted as the degree to
which $\e{a}$ and $\e{b}$ are $R$-related. Thus, a binary $\L$-relation
on $M$ may be seen as an $\L$-set in $M \times M$.
If a symbol like $\ineq$ denotes a binary $\L$-relation,
we use the usual infix notation and write $\e{a} \ineq \e{b}$
instead of ${\ineq}(\e{a},\e{b})$.

For $\L$-sets $A_1$ and $A_2$ in $M$, we put $A_1 \subseteq A_2$
whenever $A_1(\e{a}) \leq A_2(\e{a})$ for all $\e{a} \in M$
and say that $A_1$ is (fully) contained in $A_2$.
Operations with $\L$-sets are defined componentwise using
operations in $\L$. For instance, if $A_1$ and $A_2$ are $\L$-sets in $M$,
then $A_1 \cap A_2$ and $A_1 \cup A_2$ denote $\L$-sets in $M$ such that
$(A_1 \cap A_2)(\e{a}) = A_1(\e{a}) \wedge A_2(\e{a})$ for each $\e{a} \in M$
and 
$(A_1 \cup A_2)(\e{a}) = A_1(\e{a}) \vee A_2(\e{a})$ for each $\e{a} \in M$,
respectively. Note that $\cap$ and $\cup$ may be used for arbitrary arguments.
That is, for $\mathcal{A} = \{A_i;\, i \in I\}$ where all $A_i$ ($i \in I$)
are $\L$-sets in $M$, we consider an $\L$-set $\bigcap\mathcal{A}$ in $M$
which may also be denoted by $\bigcap_{i \in I}A_i$ so that
\begin{align*}
  \textstyle
  (\bigcap\mathcal{A})(\e{a}) =
  \bigl(\bigcap_{i \in I} A_i\bigr)(\e{a}) =
  \bigwedge_{i \in I}A_i(\e{a})
\end{align*}
for each $\e{a} \in M$. Analogously for $\bigcup$ and $\bigvee$.

\subsection{Algebras with Fuzzy Order}\label{ssec:awfo}
The inequalities we consider as formulas are interpreted in structures
called alegbras with fuzzy order. These structures represent graded
generalizations of the classic ordered algebras. In this section,
we recall algebras with fuzzy order and present their basic properties
which are needed to establish the completeness theorem. Details on algebraic
properties of the structures can be found in~\cite{Vy:Awfo}.

Recall that a \emph{type} of algebras is given by a set $F$ of function symbols
$f \in F$ together with their arities. We assume that the arity of
each $f \in F$ is finite. An algebra (of type $F$, see~\cite{BuSa:CUA}) is
a structure $\st = \alg$ where $M$ is
a non-empty universe set and $F^{\st}$ is a set of functions
interpreting the function symbols in~$F$. That is, for each $n$-ary
$f \in F$ there is $f^{\st} \in F^{\st}$ which is a function
$f^{\st}\!: M^n \to M$.

Let $\L$ be a complete residuated lattice.
An \emph{algebra with fuzzy order}~\cite[Definition 1]{Vy:Awfo}
(of type $F$) considering $\L$ as the structure of degrees
(shortly, an \emph{algebra with $\L$-order}) is a structure
$\st = \loalg$ such that $\alg$ is an algebra (of type $F$) and
$\ineq^{\st}$ is a binary $\L$-relation on $M$ satisfying the
following conditions:
\begin{align}
  \e{a} \ineq^{\st} \e{b} = 
  \e{b} \ineq^{\st} \e{a} = 1
  &\text{ if{}f }
  \e{a} = \e{b},
  \label{eqn:iRef} \\
  \e{a} \ineq^{\st} \e{b} \otimes \e{b} \ineq^{\st} \e{c}
  &\leq \e{a} \ineq^{\st} \e{c}, \label{eqn:iTra} \\
  \e{a}_1 \ineq^{\st} \e{b}_1 \otimes \cdots \otimes
  \e{a}_n \ineq^{\st} \e{b}_n
  &\leq f^{\st}(\e{a}_1,\ldots,\e{a}_n) \ineq^{\st}
  f^{\st}(\e{b}_1,\ldots,\e{b}_n),
  \label{eqn:iComF}
\end{align}
for all $\e{a},\e{b},\e{c},\e{a}_1,\e{b}_1,\ldots,\e{a}_n,\e{b}_n \in M$
and any $n$-ary $f \in F$.

\begin{remark}
  (a)
  Algebras with $\L$-order are generalizations of the ordinary ordered
  algebras in the following sense: If $\L$ is the two-element Boolean algebra,
  then~\eqref{eqn:iRef} yields that $\ineq^{\st}$ is a reflexive and
  antisymmetric binary relation on $M$. Moreover, \eqref{eqn:iTra} yields
  that $\ineq^{\st}$ is transitive and \eqref{eqn:iComF} is the compatibility
  condition, saying that a function in $\st$ is compatible with $\ineq^{\st}$.
  Thus, setting $\L$ to the two-element Boolean algebra,
  algebras with $\L$-orders become the ordinary ordered algebras.

  (b)
  Note that both \eqref{eqn:iTra} and \eqref{eqn:iComF} involve $\otimes$,
  i.e., the conditions of transitivity and compatibility of $\ineq^{\st}$ with
  the functions in $\st$ are formulated in terms of the
  multiplication $\otimes$ in $\L$. Condition~\eqref{eqn:iRef}
  ensures that the symmetric interior of $\ineq^{\st}$ is
  a compatible fuzzy equality relation, see~\cite[Theorem 3]{Vy:Awfo}.

  (c)
  For readers familiar with fuzzy order relations: $\ineq^{\st}$ is
  an $\L$-order in sense of~\cite[Section 4.3.1]{Bel:FRS}, i.e., it is
  $\wedge$-antisymmetric with respect to a fuzzy equality relation which
  in our case coincides with the symmetric interior of $\ineq^{\st}$.
  There are other definitions of fuzzy orders which we do not consider
  in this paper, e.g., fuzzy orders which are $\otimes$-antisymmetric
  with respect to a given similarity relation, cf.~\cite{BoDeFo:Cfwo}.
  A modestly interesting open problem is whether the subsequent results
  can be established for such alternative fuzzy orders.
\end{remark}

In our considerations on algebras with fuzzy orders, we utilize homomorphisms
and factor algebras with fuzzy orders~\cite{Vy:Awfo}. The notions are introduced
as follows. Let $\st$ and $\st[N]$ be algebras with $\L$-orders
(of the same type $F$). A map $h\!: M \to N$ which satisfies equality 
\begin{align}
  h\bigl(f^{\st}(\e{a}_1,\ldots,\e{a}_n)\bigr) &=
  f^{\st[N]}\bigl(h(\e{a}_1),\ldots,h(\e{a}_n)\bigr)
  \label{eqn:morff}
\end{align}
for any $n$-ary $f \in F$ and all $\e{a}_1,\ldots,\e{a}_n \in M$; and 
\begin{align}
  \e{a} \ineq^{\st} \e{b} &\leq h(\e{a}) \ineq^{\st[N]} h(\e{b})
  \label{eqn:morf}
\end{align}
for all $\e{a},\e{b} \in M$
is called a \emph{homomorphism}~\cite[Section 5]{Vy:Awfo} and
is denoted by $h\!: \st \to \st[N]$.
Therefore, homomorphisms are maps which are compatible with the functional
parts of $\st$ and $\st[N]$ and the $\L$-orders of $\st$ and $\st[N]$.
If $h\!: \st \to \st[N]$ is surjective, then $\st[N]$ is called
a (\emph{homomorphic}) \emph{image} of $\st$.

Consider an algebra $\st$ with $\L$-order.
A binary $\L$-relation $\qc$ on $M$ is called
an \emph{$\L$-preorder compatible with $\st$}~\cite[Section 5]{Vy:Awfo}
whenever it satisfies
\begin{align}
  \ineq^{\st} &\subseteq \qc,
  \label{eqn:qcRef} \\
  \qc(\e{a},\e{b}) \otimes \qc(\e{b},\e{c}) &\leq \qc(\e{a},\e{c}),
  \label{eqn:qcTra} \\
  \qc(\e{a}_1,\e{b}_1) \otimes \cdots \otimes
  \qc(\e{a}_n,\e{b}_n)
  &\leq
  \qc\bigl(f^{\st}(\e{a}_1,\ldots,\e{a}_n),
  f^{\st}(\e{b}_1,\ldots,\e{b}_n)\bigr),
  \label{eqn:qcCom}
\end{align}
for all $\e{a},\e{b},\e{c},\e{a}_1,\e{b}_1,\ldots,\e{a}_n,\e{b}_n \in M$
and any $n$-ary $f \in F$. Given $\st$ and an $\L$-preorder $\qc$
compatible with $\st$, we put
\begin{itemize}
\item
  $M/\qc = \bigl\{\cl[\qc]{\e{a}};\, \e{a} \in M\bigr\}$ where
  $\cl[\qc]{\e{a}} = \{\e{b} \in M;\, \qc(\e{a},\e{b}) = \qc(\e{b},\e{a}) = 1\}$;
\item
  $f^{\st/\qc}\bigl(\cl[\qc]{\e{a}_1},\ldots,\cl[\qc]{\e{a}_n}\bigr) =
  \cl[\qc]{f^{\st}(\e{a}_1,\ldots,\e{a}_n)}$;
\item
  $\cl[\qc]{\e{a}} \ineq^{\st/\qc} \cl[\qc]{\e{b}} = \qc(\e{a},\e{b})$;
\end{itemize}
and call $\st/\qc = \langle M/\qc,\ineq^{\st/\qc},F^{\st/\qc}\rangle$ the
\emph{factor algebra with $\L$-order}~\cite[Section 5]{Vy:Awfo} of $\st$ modulo $\qc$.
One can show that factor algebras with $\L$-orders are well defined algebras
with $\L$-orders, see~\cite[Lemma 4]{Vy:Awfo} for details.

The notions of homomorphic images and factor algebras preserve the desirable
properties of their classic counterparts~\cite{BuSa:CUA}. Namely, isomorphic
copies of factor algebras can be seen as representations of homomorphic images.
Indeed, for an $\L$-preorder $\qc$ which is compatible with $\st$, we may
introduce a surjective map $h_{\qc}\!: M \to M/\qc$ by putting
\begin{align}
  h_{\qc}(\e{a}) &= \cl{\e{a}}.
  \label{eqn:nat}
\end{align}
The map is called a \emph{natural homomorphism}~\cite[Section 5]{Vy:Awfo}
induced by $\qc$.
Conversely, for a surjective homomorphism $h\!: \st \to \st[N]$,
we may introduce a binary $\L$-relation $\qc_h$ on $M$ by putting
\begin{align}
  \qc_h(\e{a},\e{b}) &= h(\e{a}) \ineq^{\st[N]} h(\e{b}),
  \label{eqn:hom_th}
\end{align}
which is a compatible $\L$-preorder on $\st$ and 
$\st/\qc_h$ is isomorphic to $\st[N]$ in terms of the isomorphism
of general $\L$-structures, see~\cite{Bel:FRS,Vy:Awfo} for details.

\section{Syntax and Semantics of Fuzzy Inequational Logic}\label{sec:synsem}
This section introduces the basic notions of fuzzy inequational logic which
is developed in Pavelka style. In Subsection~\ref{ssec:sement}, we introduce
formulas, their interpretation in algebras with fuzzy orders, and present
some observations which are consequences of the general Pavelka framework.
In Subsection~\ref{ssec:synent}, we introduce a deductive system.

\subsection{Formulas, Models, and Semantic Entailment}\label{ssec:sement}
We consider formulas as syntactic expressions written in a particular language.
Namely, a \emph{language} is defined by a type $F$ of algebras (i.e., by the
collection of \emph{function symbols} with their arities,
cf. Subsection~\ref{ssec:awfo}) and a set $X$ of \emph{object variables.}
The object variables play the same role as
in predicate logics. At this point, we make no assumption on $X$. Furthermore,
the language contains the symbol $\ineq$ which is the only
\emph{relation symbol} in the language
and auxiliary symbols like parentheses and commas.

We consider the usual notion of a \emph{term}: Given $F$ and $X$, each variable
$x \in X$ is a term and if $t_1,\ldots,t_n$ are terms and $f \in F$ is
an $n$-ary function symbol, then $f(t_1,\ldots,t_n)$ is a term.
The set of all terms is then denoted $T_F(X)$ or simply $T(X)$ if $F$
is clear from the context.

A \emph{formula} (in the language given by $F$ and $X$) is any expression
\begin{align}
  t \ineq t'
  \label{eqn:ineq}
\end{align}
where $t,t' \in T_F(X)$ and it is called an \emph{inequality.}

Thus, the notion of inequality is the same as in the case of the classic
ordered algebras. For convenience, we may identify formulas with pairs
of terms in $T_F(X)$ and thus the Cartesian product $T_F(X) \times T_F(X)$
represents the set of all formulas in question. Indeed, each \eqref{eqn:ineq}
may be understood as
\begin{align}
  \langle t,t'\rangle \in T_F(X) \times T_F(X)
  \label{eqn:pair}
\end{align}
and \emph{vice versa.} Note that considering formulas as pairs of terms
like~\eqref{eqn:pair} is consistent with the abstract Pavelka approach where
formulas are supposed to be abstract objects coming from a predefined set
of all formulas which is in our case $T_F(X) \times T_F(X)$.
Therefore, we put
\begin{align}
  \Fml &= T_F(X) \times T_F(X)
\end{align}
and call $\Fml$ the \emph{set of all formulas.}

\begin{remark}
  (a)
  Let us note that in order to be able to consider any formulas, $T_F(X)$
  must be non-empty. Note that $T_F(X) \ne \emptyset$ whenever $X$
  is non-empty or $F$ contains nullary function symbols, i.e.,
  symbols for \emph{object constants.}

  (b)
  Analogously as for the classic algebras, for any complete residuated
  lattice $\L$, we may consider a term algebra with
  $\L$-order~\cite[Example 2]{Vy:Awfo}.
  Namely, if $T_F(X) \ne \emptyset$, we denote by $\TX$ the algebra
  $\langle T_F(X),\ineq^{\TX},F^{\TX}\rangle$ with $\L$-order 
  where $\ineq^{\st}$ is the identity, i.e.,
  \begin{align*}
    t \ineq^{\TX} t' &=
    \left\{
      \begin{array}{@{\,}l@{\quad}l@{}}
      1, & \text{if } t = t', \\
      0, & \text{otherwise,}
    \end{array}
    \right.
  \end{align*}
  for all $t,t' \in T_F(X)$. Furthermore, each $f^{\TX}$ is defined by
  \begin{align*}
    f^{\TX}(t_1,\ldots,t_n) &= f(t_1,\ldots,t_n).
  \end{align*}
  Thus, $\TX$ results from an ordinary
  term algebra by adding $\ineq^{\TX}$.
  $\mathbf{L}$-relations on $\TX$. We call $\TX$
  the (\emph{absolutely free}) \emph{term algebra with $\L$-order}
  over variables in $X$.
\end{remark}

We now introduce the abstract semantics for our formulas. Recall that
in the abstract Pavelka setting~\cite{Haj:MFL}
an \emph{$\L$-semantics for $\Fml$}
is a set $\mathcal{S}$ of $\L$-sets in $\Fml$. Thus, each $E \in \mathcal{S}$
is a map $E\!: \Fml \to L$ which defines for each $\phi \in \Fml$
a degree $E(\phi) \in L$ called the degree to which $\phi$ is true in $E$.
In case of our logic, we introduce $\mathcal{S}$ by evaluating inequalities
in algebras with fuzzy orders. The details are summarized below.

Let $F$, $X$, and $\L$ be fixed.
For an algebra $\st$ with $\L$-order of type $F$,
any map $v\!: X \to M$ is called
an $\st$-valuation of variables in $X$, i.e., the result $v(x)$ is the
value of $x$ in $\st$ under $v$. As usual, for each term $t \in T_F(X)$,
we define the value $\val{t}$ of $t$ in $\st$ under $v$ as follows:
\begin{align}
  \val{t} &=
  \left\{
    \begin{array}{@{\,}l@{\quad}l@{}}
      v(x), & \text{if } t \text{ is } x \in X, \\
      f^{\st}\bigl(\val{t_1},\ldots,\val{t_n}\bigr),
      & \text{if } t \text{ is } f(t_1,\ldots,t_n).
    \end{array}
  \right.
  \label{eqn:val_t}
\end{align}
Note that the usual algebraic view of~\eqref{eqn:val_t} is that the values
of terms in $\st$ under $v$ are values of homomorphisms from $\TX$ to $\st$.
Indeed, as in the classic setting, an $\st$-valuation $v\!: X \to M$
admits a unique \emph{homomorphic extension} $\home{v}\!: \TX \to \st$
for which
\begin{align}
  \home{v}(t) = \val{t}.
  \label{eqn:home_val}
\end{align}
for all $t \in T_F(X)$.
Now, for any inequality $t \ineq t'$, we may introduce
the \emph{degree to which $t \ineq t'$ is true in $\st$ under $v$} by
\begin{align}
  \val{t \ineq t'} &= \val{t} \ineq^{\st} \val{t'}.
  \label{eqn:valMv}
\end{align}
Observe that utilizing \eqref{eqn:hom_th} and \eqref{eqn:home_val},
we rewrite \eqref{eqn:valMv} as
\begin{align}
  \val{t \ineq t'} &=
  \home{v}(t) \ineq^{\st} \home{v}(t') =
  \qc_{\home{v}}(t,t'),
  \label{eqn:valMv_qc}
\end{align}
where $\qc_{\home{v}}$ is the compatible $\L$-preorder on $\TX$ induced
by the homomorphic extension $\home{v}$ of $v$. By considering the infimum
of all degrees~\eqref{eqn:valMv} ranging over all possible $\st$-valuations,
we define
\begin{align}
  \val[\st]{t \ineq t'} &= \textstyle\bigwedge_{v: X \to M}\val{t \ineq t'}
  \label{eqn:valM}
\end{align}
which is called the \emph{degree to which $t \ineq t'$ is true in $\st$}
(under all $\st$-valuations). Since we assume that $\L$ is a complete lattice,
\eqref{eqn:valM} is always defined. Utilizing~\eqref{eqn:valMv_qc}, we may
rewrite \eqref{eqn:valM} as
\begin{align}
  \val[\st]{t \ineq t'} &= 
  \textstyle\bigwedge_{v: X \to M}\qc_{\home{v}}(t,t') = 
  \bigl(\textstyle\bigcap_{v: X \to M}\qc_{\home{v}}\bigr)(t,t').
  \label{eqn:valM_qc}
\end{align}
Thus, taking into account the fact that the set of all compatible
$\L$-preorders on any algebra with $\L$-order is closed under arbitrary
intersections~\cite{Vy:Awfo}, we may consider a compatible
$\L$-preorder $\qc_{\st}$ which is defined as the intersection
of $\qc_{\home{v}}$ for all possible $\st$-valuations. That is,
\begin{align}
  \qc_{\st} &= \textstyle\bigcap_{v: X \to M}\qc_{\home{v}}.
  \label{eqn:qc_st_bigcap}
\end{align}
Under this notation, we have
\begin{align}
  \val[\st]{t \ineq t'} &= \qc_{\st}(t,t').
  \label{eqn:valst_qcst}
\end{align}
Therefore, $\qc_{\st}$ can be seen as an algebraic representation of
the degrees to which formulas are true in a given algebra $\st$ with
$\L$-order. Using this concept, we introduce the abstract semantics for
our logic in Pavelka style as follows:
\begin{align}
  \mathcal{S} &=
  \bigl\{\qc_{\st};\, \st \text{ is algebra with $\L$-order of type } F\bigr\}.
\end{align}

Now, having defined the formulas and their $\L$-semantics, the abstract
Pavelka framework gives us the notions of models and semantic entailment:
Let $\Sigma\!: \Fml \to L$, i.e., $\Sigma$ is an $\L$-set in $\Fml$
and let $\qc_{\st} \in \mathcal{S}$. Under this notation, $\qc_{\st}$
is called an \emph{$\mathcal{S}$-model of $\Sigma$} (shortly, a model)
whenever $\Sigma \subseteq \qc_{\st}$.
The set of all models of $\Sigma$ is denoted by $\Mod$. That is,
using \eqref{eqn:valst_qcst}, we have
\begin{align}
  \Mod &= \bigl\{\qc_{\st};\, \Sigma(t, t') \leq \val[\st]{t \ineq t'}
  \text{ for all } t,t' \in T_F(X)\bigr\}.
  \label{eqn:Mod}
\end{align}
Notice that $\Mod$ is indeed a set (not a proper class) which is
a subset of $\mathcal{S}$. Moreover,
the \emph{degree to which $t \ineq t'$ semantically follows by $\Sigma$}
is defined by
\begin{align}
  \val[\Sigma]{t \ineq t'} &=
  \bigl(\textstyle\bigcap\Mod\bigr)(t,t')
  \label{eqn:sement}
\end{align}
which by~\eqref{eqn:valst_qcst} and \eqref{eqn:Mod}
can be rewritten as
\begin{align}
  \val[\Sigma]{t \ineq t'} &=
  \textstyle\bigwedge\bigl\{\val[\st]{t \ineq t'}\!;\,
  \qc_{\st} \text{ is a model of } \Sigma\bigr\},
\end{align}
i.e., $\val[\Sigma]{t \ineq t'}$ is the infimum of degrees to which 
$t \ineq t'$ is true in all models of~$\Sigma$ which is the usual way
of defining degrees of semantic entailment in truth-functional logics
using (subclasses of) residuated lattices as the structures of truth
degrees.

\begin{remark}
  Note that the mainstream approach in fuzzy logics in 
  the narrow sense~\cite{EsGo:MTL,Got:Mfl,Haj:MFL}
  considers theories, i.e., the collections of formulas from which we draw
  consequences, as ordinary sets of formulas,
  cf.~\cite{CiHaNo1,CiHaNo2} covering recent results.
  In the Pavelka approach, we consider $\L$-sets of
  formulas prescribing degrees to which formulas are satisfied in models, i.e.,
  not just degrees $0$ and $1$ as in the mainstream approach. In our case,
  for each $t,t' \in T(X)$, an $\L$-set $\Sigma\!: \Fml \to L$ prescribes
  a degree $\Sigma(t, t')$ which can be interpreted as a lower bound of
  a degree to which $t \ineq t'$ shall be satisfied in a model. Clearly,
  the standard understanding of theories as sets of formulas can be viewed
  as a particular case of the concept of theories as $\L$-sets of formulas
  since $\Sigma(t, t') = 1$ prescribes that $t \ineq t'$ shall be satisfied
  (fully) in a model of $\Sigma$ and $\Sigma(t, t') = 0$ means that
  in a model of $\Sigma$ the inequality $t \ineq t'$ need not be satisfied
  at all. On the other hand, one can achieve the same goal by considering
  theories as sets of formulas and introducing formulas of the form
  $\overline{a} \Rightarrow t \ineq t'$, where $\overline{a}$ is (a constant for)
  a truth degree $a \in L$ (interpreted by the truth degree itself),
  and $\Rightarrow$ is (the symbol for) implication which is interpreted
  by $\rightarrow$ in $\L$. This approach is used by H\'ajek in his
  Rational Pavelka Logic~\cite{Haj:Flah} which extends the \plL ukasiewicz
  logic by constants for rational truth degrees in the unit interval
  and bookkeeping axioms, see also~\cite{EsGoNo:Epltc}.
  In our paper, we keep the original Pavelka approach.
\end{remark}

\subsection{Proofs and Provability Degrees}\label{ssec:synent}
We characterize the degrees of semantic entailment of ineqaulites
introduced in~\eqref{eqn:sement} by suitably defined degrees of provability.
In this subsection, we introduce a deductive system for our logic and the
next section shows its completeness in Pavelka style. We use a notation
which is close to that in \cite[Section 9.2]{Haj:MFL}.

Let us recall that deduction rules in Pavelka style can be seen
as inference rules of the form
\begin{align}
  \cfrac{\langle \phi_1,a_1\rangle,\ldots,\langle \phi_n,a_n\rangle}{
    \langle \psi,b\rangle}\,,
  \label{eqn:dedrule}
\end{align}
where $\phi_1,\ldots,\phi_n,\psi$ are formulas and $a_1,\ldots,a_n,b$
are degrees in $\L$. The rule~\eqref{eqn:dedrule} reads:
``from $\phi_1$ valid to degree $a_1$ and $\cdots$ and $\phi_n$
valid to degree $a_n$, infer $\psi$ valid to degree $b$''. Hence,
unlike the ordinary deduction rules which only have the syntactic
component which in our case says that $\psi$ is derived
from $\phi_1,\ldots,\phi_n$, the rule \eqref{eqn:dedrule} has
an additional semantic component which computes the degree $b$ based
on the degrees $a_1,\ldots,a_n$.

Formally, an $n$-ary \emph{deduction rule} is a pair
$R = \langle R_1,R_2\rangle$ where $R_1$, called the
\emph{syntactic part of $R$},
is a~partial map from $\Fml^n$ to $\Fml$ and $R_2$, called the
\emph{semantic part of $R$} is a map $R_2\!: L^n \to L$.
A rule $R = \langle R_1,R_2\rangle$ such that
$R_1(\phi_1,\ldots,\phi_n) = \psi$ and $R_2(a_1,\ldots,a_n) = b$
is usually depicted as in~\eqref{eqn:dedrule}.
The semantic part $R_2$ of an $n$-ary deduction rule
$R=\langle R_1,R_2\rangle$ \emph{preserves non-empty suprema} if
\begin{align}
  \textstyle R_2(\ldots,\bigvee_{\!i \in I}a_i,\dots) &=
  \textstyle \bigvee_{\!i \in I}R_2(\ldots,a_i,\dots)
\end{align}
for each $I \ne \emptyset$ and $a_i \in L$ $(i \in I)$.
A \emph{deductive system} for $\mathit{Fml}$ and $\mathbf{L}$
is a pair $\langle A, \mathcal{R}\rangle$, where
\begin{enumerate}
\item[\itm{1}]
  $A\!: \Fml \to L$ is an $\L$-set of \emph{axioms}, and
\item[\itm{2}]
  $\mathcal{R}$ is a set of deduction rules,
  each preserving non-empty suprema.
\end{enumerate}
In our logic, we use a concrete deductive system $\langle A, \mathcal{R}\rangle$
where the $\L$-set $A$ of axioms is defined by
\begin{align}
  A(t, t') &=
  \left\{
    \begin{array}{@{\,}l@{\quad}l@{}}
      1, &\text{if } t = t', \\
      0, &\text{otherwise,}
    \end{array}
  \right.
  \label{eqn:Ax}
\end{align}
and $\mathcal{R}$ consists of the following deduction rules:
\begin{align}
  \mathrm{Tra}\!:\, &\cfrac{\langle t \ineq t', a\rangle,
    \langle t' \ineq t'', b\rangle}{
    \langle t \ineq t'', a \otimes b\rangle},
  \label{eqn:Tra} \\
  \mathrm{Com}\!:\, &\cfrac{\langle t_1 \ineq t'_1, a_1\rangle,\ldots,
    \langle t_n \ineq t'_n, a_n\rangle}{
    \langle f(t_1,\ldots,t_n) \ineq f(t'_1,\ldots,t'_n),
    a_1 \otimes \cdots \otimes a_n\rangle},
  \label{eqn:Com} \\
  \mathrm{Inv}\!:\, &\cfrac{\langle t \ineq t', a\rangle}{
    \langle h(t) \ineq h(t'), a\rangle}, 
  \label{eqn:Inv}
\end{align}
where $t,t',t'',t_1,t'_1,\ldots,t_n,t'_n \in T_F(X)$, 
$f$ is an $n$-ary function symbol in $F$,
$h$ is a homomorphism $h\!: \TX \to \TX$, and
$a,b,a_1,\ldots,a_n \in L$. The rules are called the rules
of \emph{transitivity}, \emph{compatibility}, and \emph{invariance,}
respectively.

\begin{remark}
  (a)
  Note that the rules of compatibility and invariance
  in~\eqref{eqn:Com} and \eqref{eqn:Inv} represent in fact multiple rules. 
  Indeed, for each function symbol $f \in F$, \eqref{eqn:Com} defines
  a separate deduction rule with the same number of input formulas as
  the arity of $f$. In the second case, for each $h$, \eqref{eqn:Inv}
  defines a separate deduction rule in sense of Pavelka. Note that all
  the rules have natural meaning. For instance, \eqref{eqn:Tra} reads:
  ``from $t \ineq t'$ valid to degree $a$ and $t' \ineq t''$ valid
  to degree $b$, infer $t \ineq t''$ valid (at least) to degree $a \otimes b$''.
  The compatibility rule can be interpreted analogously. The rule of invariance
  represents a particular substitution rule when from $t \ineq t'$ valid to
  degree $a$ one infers inequality $h(t) \ineq h(t')$ valid at least
  to degree $a$. Observe that $h(t)$ represents the result of a simultaneous
  substitution of each variable $x$ in term $t$ by term $h(x)$.

  (b)
  All the rules \eqref{eqn:Tra}--\eqref{eqn:Inv} preserve non-empty suprema
  since as a consequence of the adjointness property of $\L$,
  $\otimes$ is distributive with respect to general suprema $\bigvee$ in $\L$,
  see~\cite[Theorem 1.22]{BeVy:FEL}.
\end{remark}

Using our deduction system, we introduce provability degrees. Recall that
in the abstract Pavelka approach, we define proofs consisting of formulas
annotated by degrees in $\L$ as follows. Let $\langle A,\mathcal{R}\rangle$
be a deductive system for $\Fml$ and $\L$ and
let $\phi \in \Fml$ and $a \in L$.
A \emph{proof} (\emph{annotated by degrees in $\L$})
of $\langle \phi,a\rangle$ by $\Sigma$ using $\langle A,\mathcal{R}\rangle$
is a sequence
\begin{align*}
  \langle \phi_1,a_1\rangle,\dots,\langle \phi_n,a_n\rangle
\end{align*}
such that $\phi_n$ is $\phi$, $a_n=a$,
and for each $i=1,\dots,n$, we have
\begin{itemize}
\item[\itm{1}]
  $a_i=\Sigma(\phi_i)$, or
\item[\itm{2}]
  $a_i=A(\phi_i)$, or
\item[\itm{3}]
  there are $\langle \phi_{j_1},a_{j_1}\rangle,\ldots,
  \langle \phi_{j_k},a_{j_k}\rangle$ such that $j_1,\ldots,j_k < i$
  and there is $\langle R_1,R_2\rangle \in \mathcal{R}$ such that 
  $\phi_i = R_1(\phi_{j_1},\ldots,\phi_{j_k})$ and
  $a_i = R_2(a_{j_1},\ldots,a_{j_k})$.
\end{itemize}
If there is a proof of $\langle \phi,a\rangle$ by $\Sigma$ using
$\langle A,\mathcal{R}\rangle$,
we write $\Sigma \vdash^{\langle A,{\cal R}\rangle} \langle \phi,a\rangle$
and call $\langle \phi,a\rangle$ \emph{provable} by $\Sigma$
using $\langle A,\mathcal{R}\rangle$.
If $\Sigma \vdash^{\langle A,{\cal R}\rangle} \langle\phi,a\rangle$,
we also call $\phi$ \emph{provable
  by $\Sigma$ using $\langle A,\mathcal{R}\rangle$
  at least to degree $a$.}

Finaly, the \emph{degree of provability} of $\phi$ by $\Sigma$
using $\langle A,\mathcal{R}\rangle$, which is denoted by
$\prv[\Sigma]{\phi}^{\langle A,\mathcal{R}\rangle}$, is defined as follows:
\begin{align}
  \prv[\Sigma]{\phi}^{\langle A,\mathcal{R}\rangle} =
  \textstyle\bigvee\bigl\{a \in L;\,
  \Sigma\vdash^{\langle A,\mathcal{R}\rangle}\langle \phi,a\rangle\bigr\}.
  \label{eqn:prv}
\end{align}
That is, $\prv[\Sigma]{\phi}^{\langle A,\mathcal{R}\rangle}$ is the supremum
of all degrees to which $\phi$ is provable by $\Sigma$. If we use our deductive
system which consists of $A$ defined by~\eqref{eqn:Ax} and
\eqref{eqn:Tra}--\eqref{eqn:Inv} as the deduction rules,
we omit the superscript $\langle A,\mathcal{R}\rangle$ and
write just $\prv[\Sigma]{\phi}$ and $\Sigma \vdash \langle \phi,a\rangle$.

\begin{remark}
  The rules of compatibility and invariance may be substituted by
  alternative deduction rules which generalize the classic rules of
  replacement and substitution often considered in universal algebra
  and inequational logic and which also appear in~\cite{Be:Fel,BeVy:FEL}.
  Namely, the rule of \emph{replacement} is
  \begin{align*}
    \mathrm{Rep}\!:\, &\cfrac{\langle t \ineq t', a\rangle}{
      \langle s \ineq s', a\rangle},
  \end{align*}
  where $s$ is a term containing $t$ as a subterm and $s'$ results by
  $s$ by replacing one occurrence of $t$ by $t'$. A moment's reflection
  shows that if \eqref{eqn:Com} derives
  $f(t_1,\ldots,t_n) \ineq f(t'_1,\ldots,t_n)$ valid to degree
  $a_1 \otimes \cdots \otimes a_n$ then the same result can be achieved
  by $n$ applications of the replacement rule which derives
  $f(t_1,\ldots,t_n) \ineq f(t'_1,t_2,\ldots,t_n)$ to degree $a_1$,
  $f(t'_1,t_2,\ldots,t_n) \ineq f(t'_1,t'_2,t_3,\ldots,t_n)$
  to degree $a_2$, and $\cdots$ and,
  $f(t'_1,\ldots,t'_{n-1},t_n) \ineq f(t'_1,\ldots,t'_n)$
  valid to degree $a_n$ followed by $n$ applications of~\eqref{eqn:Tra}.
  Conversely, by induction over the rank of $s$ one can show
  that utilizing the axioms~\eqref{eqn:Ax}, one can produce the result of
  the replacement rule by applying~\eqref{eqn:Com}.

  The rule of \emph{substitution} is 
  \begin{align*}
    \mathrm{Sub}\!:\, &\cfrac{\langle t \ineq t', a\rangle}{
      \langle t(x/s) \ineq t'(x/s), a\rangle},
  \end{align*}
  where $t(x/s)$ and $t'(x/s)$ denote terms which result by $t$ and $t'$
  by substituting the term $s$ for each occurrence of the variable $x$ in
  $t$ and $t'$, respectively. Clearly, the rule of substitution is a particular
  case~\eqref{eqn:Inv}. On the other hand, if $X$ is denumerable,
  one may obtain the general result of~\eqref{eqn:Inv} by a series of
  applications of the rule of substitution. Let us note that in order to
  correctly implement the simultaneous substitution of~\eqref{eqn:Inv},
  one has to first substitute all variables in $t$ and $t'$ by
  variables which do not appear in either of $t$, $t'$, and $s$
  and thus the assumption on $X$ being (at least) denumerable is essential.
\end{remark}

Let us note here that the degrees of semantic entailment and the provability
degrees indroduced in this section generalize the classic concepts of
semantic entailment and provability in the following sense:
If $\Sigma$ is a crisp $\L$-set, i.e., if $\Sigma(t,t') \in \{0,1\}$
for all $t,t' \in T_F(X)$, then $\Sigma$ may be seen as an ordinary subset
of $\Fml$. In addition, $\val[\Sigma]{t \ineq t'} \in \{0,1\}$
and $\val[\Sigma]{t \ineq t'} = 1$ if{}f $t \ineq t'$ follows by $\Sigma$
in the usual sense (i.e., if{}f $t \ineq t'$ is true in each ordered
algebra which is a model of $\Sigma$, see the proof
of \cite[Theorem 12]{Vy:Awfo} for details).
Analogously,
$\prv[\Sigma]{t \ineq t'} \in \{0,1\}$ and $\prv[\Sigma]{t \ineq t'} = 1$
if{}f $t \ineq t'$ is provable by $\Sigma$ in the usual sense
(i.e., if{}f $t \ineq t'$ is provable by $\Sigma$ using the inference
system of the classic inequational logic). This situation occurs in
particular if $\L$ is the two-element Boolean algebra. Therefore,
the graded concepts of semantic and syntactic entailment in the
Pavelka approach is what makes our logic non-trivial.

\section{Completeness of Fuzzy Inequational Logic}\label{sec:compl}
In this section, we show that our logic is Pavelka-style complete over any $\L$.
It means that the degrees of semantic entailment agree with the degrees of
provability. Thus, for each $\Sigma$ and $t \ineq t'$, we establish
$\prv[\Sigma]{t \ineq t'} = \val[\Sigma]{t \ineq t'}$. Note that the
$\leq$-part of the claim (Pavelka-style soundness) is more or less evident.
We establish the equality by proving that the semantic and syntactic
closures associated to any $\L$-set of formulas coincide.
Some properties of the closures
and their relationship to the degrees of semantic entailment and provability
follow directly from properties of the abstract Pavelka framework.

We say that $\Sigma\!: \Fml \to L$ is \emph{semantically closed}
whenever
\begin{align}
  \val[\Sigma]{t \ineq t'} \leq \Sigma(t, t')
  \label{eqn:semclos}
\end{align}
for all $t,t' \in T_F(X)$.
Since the converse inequality always holds, $\Sigma$ is semantically closed if{}f
$\val[\Sigma]{t \ineq t'} = \Sigma(t, t')$ for all $t,t' \in T_F(X)$. Note that
using~\eqref{eqn:sement}, $\Sigma$ is semantically closed if{}f
$\Sigma = \textstyle\bigcap\Mod$.

As an immediate consequence, we get that the set of all semantically
closed $\L$-sets of formulas forms a closure system. In order to see that,
observe that if $\textstyle\bigcap\Mod[\Sigma_i] \subseteq \Sigma_i$ ($i \in I$)
then $\textstyle\bigcap_{i \in I}\bigcap\Mod[\Sigma_i] \subseteq
\bigcap_{i \in I}\Sigma_i$. Furthermore,
$\bigcap_{i \in I}\Sigma_i \subseteq \Sigma_i$ yields
$\bigcap\Mod[\bigcap_{i \in I}\Sigma_i] \subseteq \bigcap\Mod[\Sigma_i]$
for all $i \in I$ and thus
\begin{align*}
  \textstyle\bigcap\Mod[\textstyle\bigcap_{i \in I}\Sigma_i] \subseteq
  \textstyle\bigcap_{i \in I}\textstyle\bigcap\Mod[\Sigma_i] \subseteq 
  \textstyle\bigcap_{i \in I}\Sigma_i,
\end{align*}
showing that $\bigcap_{i \in I}\Sigma_i$ is semantically closed.
We may therefore consider the \emph{semantic closure $\semcl$} of $\Sigma$,
i.e., $\semcl$ is the least semantically closed set of formulas
containing $\Sigma$:
\begin{align}
  \semcl &= \textstyle\bigcap\{\Sigma';\,
  \Sigma \subseteq \Sigma' \text{ and } \textstyle\bigcap\Mod[\Sigma'] \subseteq \Sigma'\}.
  \label{eqn:semcl}
\end{align}
The semantic closure $\semcl$ of $\Sigma$ determines the degrees of semantic entailment.
Indeed, $\Sigma \subseteq \semcl$ yields
$\bigcap\Mod \subseteq \bigcap\Mod[\semcl] \subseteq \semcl$. Moreover,
$\bigcap\Mod$ is semantically closed owing to
the fact that $\Mod \subseteq \Mod[\bigcap\Mod]$ which is
easy to see since $\qc_{\st} \in \Mod$ implies
$\bigcap\Mod \subseteq \qc_{\st}$ and so $\qc_{\st} \in \Mod[\bigcap\Mod]$.
Altogether, we get that
\begin{align}
  \semcl &= \textstyle\bigcap\Mod
  \label{eqn:semcl_bigcap_Mod}
\end{align}
which using~\eqref{eqn:sement} means that
\begin{align}
  \semcl(t,t') &= \val[\Sigma]{t \ineq t'}
  \label{eqn:semcl=sement}
\end{align}
for all $t,t' \in T_F(X)$. Since $\semcl = \semcl[(\semcl)]$, we further derive
\begin{align}
  \val[\Sigma]{t \ineq t'} &= \semcl(t,t') = \val[\semcl]{t \ineq t'}
  \label{eqn:semcl_idemp}
\end{align}

Recall that
$\Sigma\!: \Fml \to L$ is called
\emph{syntactically closed} under $\langle A,\mathcal{R}\rangle$
if $A \subseteq \Sigma$ and for any $n$-ary deduction rule
$\langle R_1,R_2\rangle$ in $\mathcal{R}$ and arbitrary
formulas $\phi_1,\dots,\phi_n$, we have
\begin{align}
  R_2(\Sigma(\phi_1),\dots,\Sigma(\phi_n)) \leq
  \Sigma(R_1(\phi_1,\dots,\phi_n))
  \label{eqn:syncl_general}
\end{align}
provided that $R_1(\phi_1,\dots,\phi_n)$ is defined. In case of the deduction
system of our logic which consists of $A$ defined by~\eqref{eqn:Ax}
and deduction rules~\eqref{eqn:Tra}--\eqref{eqn:Inv},
the previous condition of $\Sigma$ being syntactically closed
translates into
\begin{align}
  \Sigma(t,t) &= 1,
  \label{eqn:dedRef} \\
  \Sigma(t,t') \otimes \Sigma(t',t'') &\leq \Sigma(t,t''),
  \label{eqn:dedTra} \\
  \Sigma(t_1,t'_1) \otimes \cdots \otimes \Sigma(t_n,t'_n)
  &\leq \Sigma(f(t_1,\ldots,t_n),f(t'_1,\ldots,t'_n)),
  \label{eqn:dedCom} \\
  \Sigma(t,t') &\leq \Sigma(h(t),h(t')),
  \label{eqn:dedInv}
\end{align}
which all must be satisfied for all
$t,t',t'',t_1,t'_1,\ldots,t_n,t'_n \in T_F(X)$, 
any $n$-ary function symbol $f \in F$, and
any homomorphism $h\!: \TX \to \TX$.

It can be shown that the set of all syntactically closed $\L$-sets of
formulas forms a closure system, see~\cite{Pav:Ofl1}
and \cite[Lemma 9.2.5]{Haj:MFL}. The \emph{syntactic closure $\syncl$}
of $\Sigma$ is thus introduced by
\begin{align}
  \syncl &= \textstyle\bigcap\{\Sigma';\,
  \Sigma \subseteq \Sigma' \text{ and }
  \Sigma' \text{ is syntactically closed}\}.
  \label{eqn:syncl}
\end{align}
As a consequence of the fact that the syntactic parts of deduction rules
in deductive systems preserve non-empty suprema, it follows that
\begin{align}
  \syncl(t,t') &= \prv[\Sigma]{t \ineq t'}
  \label{eqn:syncl=provab}
\end{align}
for all $t,t' \in T_F(X)$,
see \cite{Pav:Ofl1} and \cite[Theorem 9.2.8]{Haj:MFL}.

We now turn our attention to the completeness of our logic. By the previous
observations on the relationship between the syntactic/semantic entailments
and syntactic/semantic closures of $\L$-sets of formulas, in order to prove
that our logic is Pavelka-style complete, it suffices to show the equality
of syntactic and semantic closures for any $\Sigma$. The proof is elaborated
by the following two lemmas.

\begin{lemma}\label{le:syncl_semcl}
  For any $\Sigma\!: \Fml \to L$, we have $\syncl \subseteq \semcl$.
\end{lemma}
\begin{proof}
  It suffices to check that $\semcl$ contains $\Sigma$ and is syntactically
  closed because $\syncl$ is the least syntactically closed $\L$-set in $\Fml$
  containing $\Sigma$.

  Obviously, $\Sigma \subseteq \semcl$ and thus it suffices to check that
  $\semcl$ satisfies all \eqref{eqn:dedRef}--\eqref{eqn:dedInv}. Trivially,
  $\semcl$ satisfies \eqref{eqn:dedRef} because $\qc_{\st}(t,t) = 1$
  for any $\qc_{\st} \in \Mod$. In order to see that \eqref{eqn:dedTra}
  is satisfied, take $t,t',t'' \in T_F(X)$ and observe that
  \eqref{eqn:qcTra},
  \eqref{eqn:qc_st_bigcap}, and
  \eqref{eqn:semcl_bigcap_Mod} together with the fact that
  $a \otimes \bigwedge_{i \in I}b_i \leq \bigwedge_{i \in I}(a \otimes b_i)$
  yield
  \begin{align*}
    \semcl(t,t') \otimes \semcl(t',t'') &=
    \bigl(\textstyle\bigcap\Mod\bigr)(t,t') \otimes
    \bigl(\textstyle\bigcap\Mod\bigr)(t',t'') \\
    &\leq
    \textstyle\bigwedge_{\qc_{\st} \in \Mod}
    \bigl(\qc_{\st}(t,t') \otimes \qc_{\st}(t',t'')\bigr) \\
    &\leq
    \textstyle\bigwedge_{\qc_{\st} \in \Mod}
    \textstyle\bigwedge_{v: X \to M}
    \bigl(\qc_{\home{v}}(t,t') \otimes \qc_{\home{v}}(t',t'')\bigr) \\
    &\leq
    \textstyle\bigwedge_{\qc_{\st} \in \Mod}
    \textstyle\bigwedge_{v: X \to M}
    \qc_{\home{v}}(t,t'') \\
    &= \semcl(t,t'').
  \end{align*}
  Analogously, one may check~\eqref{eqn:dedCom} utilizing \eqref{eqn:qcCom}.
  Finally, \eqref{eqn:dedInv} is satisfied because for every
  homomorphism $h\!: \TX \to \TX$ and $\st$-valuation $v\!: X \to M$ one can take
  an $\st$-valuation $w\!: X \to M$ satisfying $w(x) = \home{v}(h(x))$
  for all $x \in X$. For $w$, by induction over the rank of terms,
  we get that $\home{w}(t) = \home{v}(h(t))$ for all $t \in T_F(X)$.
  Therefore,
  \begin{align*}
    \semcl(t,t') &=
    \textstyle\bigwedge_{\qc_{\st} \in \Mod}
    \textstyle\bigwedge_{w: X \to M}
    \qc_{\home{w}}(t,t') \\
    &\leq \textstyle\bigwedge_{\qc_{\st} \in \Mod}
    \textstyle\bigwedge_{v: X \to M}
    \qc_{\home{v}}(h(t),h(t')) \\
    &= \semcl(h(t),h(t')).
  \end{align*}
  Therefore, $\semcl$ is syntactically closed.
\end{proof}

Note that using~\eqref{eqn:semcl=sement}, \eqref{eqn:syncl=provab},
and Lemma~\ref{le:syncl_semcl}, we get that our logic is sound:
\begin{align}
  \prv{t \ineq t'} =
  \syncl(t,t') \leq
  \semcl(t,t') =
  \val[\Sigma]{t \ineq t'}.
\end{align}
The next lemma proves the converse inequality.

\begin{lemma}\label{le:semcl_syncl}
  For any $\Sigma\!: \Fml \to L$, we have $\semcl \subseteq \syncl$.
\end{lemma}
\begin{proof}
  It suffices to check that $\syncl$ contains $\Sigma$ and is semantically
  closed because $\semcl$ is the least semantically closed $\L$-set in $\Fml$
  containing $\Sigma$.

  Observe that since $\syncl$ satisfies~\eqref{eqn:dedRef}--\eqref{eqn:dedCom},
  it is a compatible $\L$-preorder on $\TX$ and by definition it
  contains $\Sigma$. Therefore, we may consider the factor
  algebra $\TX/\syncl$ with $\L$-order. For the
  factor algebra we now prove that $\syncl = \qc_{\TX/\syncl}$ by checking both
  inclusions.

  Take a $\TX/\syncl$-valuation $v\!: X \to T_F(X)/\syncl$ such that
  $v(x) = \cl[\syncl]{x}$. For its homomorphic extension
  $\home{v}\!: \TX \to \TX/\syncl$, we have $\home{v}(t) = \cl[\syncl]{t}$
  for all $t \in T_F(X)$. As a consequence
  \begin{align*}
    \qc_{\TX/\syncl}(t,t') \leq
    \qc_{\home{v}}(t,t') =
    \cl[\syncl]{t} \ineq^{\TX/\syncl} \cl[\syncl]{t'} =
    \syncl(t,t'),
  \end{align*}
  which proves that $\qc_{\TX/\syncl} \subseteq \syncl$. Conversely, take
  $v\!: X \to T_F(X)/\syncl$ and let $h\!: X \to T_F(X)$ be a map such that
  $h(x) \in v(x)$ for all $x \in X$. For the homomorphic extension $\home{h}$
  of $h$, we get $\home{v}(t) = \cl[\syncl]{\home{h}(t)}$
  for all $t \in T_F(X)$. As a consequence of~\eqref{eqn:dedInv},
  it follows that
  \begin{align*}
    \syncl(t,t') \leq
    \syncl(\home{h}(t),\home{h}(t')) =
    \cl[\syncl]{\home{h}(t)} \ineq^{\TX/\syncl} \cl[\syncl]{\home{h}(t')} =
    \qc_{\home{v}}(t,t'),
  \end{align*}
  showing $\syncl \subseteq \qc_{\home{v}}$. Since $v$ is arbitrary,
  we get $\syncl \subseteq \qc_{\TX/\syncl}$.

  We now finish the proof as follows.
  Using the inclusion $\syncl \subseteq \qc_{\TX/\syncl}$,
  we get $\qc_{\TX/\syncl} \in \Mod[\syncl]$ and thus
  $\textstyle\bigcap \Mod[\syncl] \subseteq \qc_{\TX/\syncl} \subseteq \syncl$
  on account of $\qc_{\TX/\syncl} \subseteq \syncl$. This proves that 
  $\syncl$ is semantically closed.
\end{proof}

To sum up, we have established the following completeness theorem:

\begin{theorem}[completeness]\label{th:compl}
  For any $\Sigma\!: \Fml \to L$ and $t,t' \in T_F(X)$, we have
  \begin{align}
    \prv{t \ineq t'} = \val[\Sigma]{t \ineq t'}.
  \end{align}
\end{theorem}
\begin{proof}
  Consequence of~\eqref{eqn:semcl=sement}, \eqref{eqn:syncl=provab},
  Lemma~\ref{le:syncl_semcl}, and Lemma~\ref{le:semcl_syncl}.
\end{proof}

We conclude the section by remarks on the completeness.

\begin{remark}
  (a)
  Our inequational logic can be seen as a particular fragment of a first-order
  fuzzy logic which only uses atomic formulas and a single relation symbol---the
  symbol for inequality. For this particular fragment, we have established
  Pavelka-style completeness over arbitrary $\L$. This is in contrast to 
  the full first-order logic (with all connectives in the language including
  the implication) where Pavelka-style completeness depends
  on the continuity of the truth functions
  of logical connectives, cf.~\cite{Haj:MFL,Pav:Ofl1,Pav:Ofl2,Pav:Ofl3}.

  (b)
  As a consequence of Theorem~\ref{th:compl}, we get that $\semcl$
  (which is equal to $\syncl$) is a compatible $\L$-preorder on $\TX$
  which in addition satisfies~\eqref{eqn:dedInv}, i.e.,
  it is a \emph{fully invariant} compatible $\L$-preorder on $\TX$
  and the factor algebra $\TX/\semcl$ with $\L$-order fully describes
  the degrees of entailment by $\Sigma$ because
  \begin{align*}
    \prv{t \ineq t'} =
    \val[\Sigma]{t \ineq t'} =
    \val[\TX/\semcl]{t \ineq t'}.
  \end{align*}
  This generalizes the well-known property of syntactically/semantically
  closed sets of inequalities in case of the classic inequational logic.

  (c)
  Also note that the notion of provability degree is not finitary in
  the usual sense: $\prv{t \ineq t'} = a$ does not guarantee that
  $\Sigma \vdash \langle t \ineq t', a\rangle$. The arguments are the
  same as in the case of the fuzzy equational logic~\cite{Be:Fel},
  cf.~\cite[Example 3.32]{BeVy:FEL}.
\end{remark}

\section{Application: Abstract Logic of Graded Attributes}\label{sec:fal}
We now show how the general result in the previous section can be used to
obtain complete axiomatizations of logics dealing with particular problem
domains. For illustration, we show a general logic for reasoning with graded
if-then rules which generalize the ordinary attribute implications which
appear in formal concept analysis~\cite{GaWi:FCA} of
relational object-attribute data. In this section, we first recall
the notions related to attribute implications and their entailment and
then we present their generalization which exploits the results from
Section~\ref{sec:compl}.

Consider a finite set $Y$ of symbols called \emph{attributes.} An attribute
implication over $Y$ is an expression
\begin{align}
  A \Rightarrow B
  \label{eqn:ai}
\end{align} 
such that $A,B \subseteq Y$. The intended meaning of $A \Rightarrow B$ is to
express a dependency ``if an object has all the attributes in $A$,
then it has all the attributes in $B$'' and if $A = \{p_1,\ldots,p_m\}$
and $B = \{q_1,\ldots,q_n\}$, the attribute implication \eqref{eqn:ai} is
written as
\begin{align}
  \{p_1,\ldots,p_m\} \Rightarrow \{q_1,\ldots,q_n\}.
  \label{eqn:ai_pq}
\end{align}
For $A,B,M \subseteq Y$, we call $A \Rightarrow B$ \emph{satisfied by $M$}
(or \emph{true in $M$}) whenever $A \subseteq M$ implies $B \subseteq M$
(i.e., $A \nsubseteq M$ or $B \subseteq M$) and denote the fact
by $M \models A \Rightarrow B$. Note that if $M$ is considered as
a set of attributes of an object, then 
$M \models A \Rightarrow B$ means that 
``If the object has all the attributes in $A$,
then it has all the attributes in $B$'' which corresponds with the
intended meaning outlined above.

\begin{remark}
  Let us note that formulas like~\eqref{eqn:ai_pq} appear in other
  disciplines and are extensively used for knowledge representation and
  reasoning about data dependencies. For instance, they are known under
  the name \emph{functional dependencies} in relational
  databases~\cite{Mai:TRD} and can be seen as particular
  \emph{definite clauses} used in logic programming~\cite{Lloyd84}.
  Interestingly, even if the database semantics of the rules differs
  from the one introuced above, it yields the same notion of
  semantic entailment~\cite{Fa:Fdrbpl} and thus a common axiomatization.
  Rules like~\eqref{eqn:ai_pq} are also used in data mining
  as \emph{association rules} \cite{AgImSw:ASR,Zak:Mnrar},
  with their validity in data being defined using constraints such as
  confidence and support.
\end{remark}

Semantic entailment of attribute implications is introduced as follows.
A set $M \subseteq Y$ is called a \emph{model} of a set $\Sigma$
of attribute implications whenever $M \models A \Rightarrow B$
for all $A \Rightarrow B \in \Sigma$. Furthermore, $A \Rightarrow B$ is
\emph{semantically entailed} by $\Sigma$, written
$\Sigma \models A \Rightarrow B$, if $M \models A \Rightarrow B$
for each model $M$ of $\Sigma$.

The semantic entailment of attribute implications has an axiomatization
which is based on the following deduction rules
\begin{align}
  \mathrm{Ax}\!:\,
  &\cfrac{}{A {\cup} B \Rightarrow A},
  &
  \mathrm{Tra}\!:\,
  &\cfrac{A \Rightarrow B, B \Rightarrow C}{A \Rightarrow C},
  &
  \mathrm{Aug}\!:\,
  &\cfrac{A \Rightarrow B}{A{\cup}C \Rightarrow B{\cup}C},
  \label{eqn:orig_Arms}
\end{align}
where $\cup$ denotes the set-theoretic union and $A,B,C \subseteq Y$.
Note that $\mathrm{Ax}$ is a nullary rule, i.e., each $A{\cup}B \Rightarrow A$
is an axiom. Using the deduction rules, we define the usual notion of
provability of attribute implications from sets of attribute implications:
for $\Sigma$ and $A \Rightarrow B$, we put $\Sigma \vdash A \Rightarrow B$
whenever there is a sequence (a proof) $\phi_1,\ldots,\phi_n$ such that
$\phi_n$ is $A \Rightarrow B$ and each $\phi_i$ in the sequence is in $\Sigma$
or results by preceding formulas in the sequence using $\mathrm{Ax}$,
$\mathrm{Tra}$, or $\mathrm{Aug}$. The usual completeness theorem
is established: $\Sigma \models A \Rightarrow B$ if{}f
$\Sigma \vdash A \Rightarrow B$.

The axiomatization based on $\mathrm{Ax}$, $\mathrm{Tra}$, and $\mathrm{Aug}$
was discovered by Armstrong~\cite{Arm:Dsdbr}. There are other equivalent
systems of deductions rules which are even simpler. For instance, 
$\mathrm{Tra}$ (transitivity), and $\mathrm{Aug}$ (augmentation)
can be equivalently replaced by the rule of \emph{cut}
(also known as \emph{pseudo-transitivity}~\cite{Mai:TRD}):
\begin{align}
  \mathrm{Cut}\!:\,
  &\cfrac{A \Rightarrow B, B{\cup}C \Rightarrow D}{A{\cup}C \Rightarrow D}
  \label{eqn:Cut_ai}
\end{align}
for all $A,B,C,D \subseteq Y$.

In this section, we propose a general form of formulas like~\eqref{eqn:ai_pq}
with general semantics and a complete Pavelka-style axiomatization.
In particular, we focus on a generalization where \emph{attributes are graded.}
That is, instead of considering the presence/absence of attributes as in the 
classic setting, we allow attributes to be \emph{present to degrees} and we 
allow graded entailment of rules from $\L$-sets of other rules, following
Pavelka's approach. The presented extension is motivated by the fact that
in many situations, a data analyst may want to express validity of rules
to degrees and may want to be able to make an \emph{approximate inference}
based on partially true rules.

\begin{remark}
  There are approaches which generalize attribute implications in
  a graded setting.
  Most notably, the early approach by Polandt~\cite{Po:FB} which introduces
  attribute implications as formulas in the formal concept analysis
  of graded object-attribute data and the more general approach by
  Belohlavek and Vychodil~\cite{BeVy:ADfDwG} which parameterizes the
  semantics of the rules by linguistic hedges~\cite{BeVy:Fcalh,EsGoNo:Hedges,Za:Afstilh}.
  The approaches are different from the generalization presented below.
  Namely, \cite{BeVy:ADfDwG} uses rules which may be seen as implications
  between graded $\L$-sets of attributes, i.e., the (constants for) truth
  degrees appear explicitly in the antecedents and consequents of the rules.
  In contrast, the generalization in this section does not use (constants for)
  truth degrees in formulas but, on the other hand, it offers a more general
  interpretation of the rules, e.g., $\Rightarrow$ may have other
  interpretations than the residuum in $\L$.
\end{remark}

We start by considering formulas of our general logic of attribute implications.
Although it is widely used, the set-theoretic treatment of attribute
implications like~\eqref{eqn:ai_pq} is somewhat limiting. For instance, it
implies that the (interpretation of) conjunction which is tacitly used in
the definition of $M \models A \Rightarrow B$ is idempotent. Of course, this
is true in the classic setting but it may not be desirable in
a graded generalization. Therefore, we view \eqref{eqn:ai_pq} as
a (propositional) formula of the form
\begin{align}
  \bigl(p_1 \logand \cdots \logand p_m \logand \top)
  \Rightarrow
  \bigl(q_1 \logand \cdots \logand q_n \logand \top\bigr),
  \label{eqn:ai_prop}
\end{align}
where $\Rightarrow$ is a symbol for material implication,
$\logand$ is a symbol for conjunction, and $\top$ is
the truth constant denoting $1$ (the truth value ``true'').
Observe that $\top$ is needed to correctly handle the case of
$m = 0$ or $n = 0$.
Thus, in the narrow sense, an attribute implication can be seen as
a (propositional) formula in the form of an implication between conjunctions
of attributes in $Y$ (which are considered as propositional variables).
Since the classic $\logand$ is commutative, associative, and idempotent,
the order of variables, additional parentheses, or duplicities of variables
may be neglected.

Formula \eqref{eqn:ai_prop} is true under a given evaluation $e$
of propositional variables in sense of the classical propositional logic,
if the value of the antecedent (under the evaluation $e$)
is \emph{less than or equal to} the value of the
consequent (under the evaluation $e$). Therefore, \eqref{eqn:ai_prop}
being true may be expressed via the \emph{ordering of truth degrees.}
The main idea of our approach is to utilize general $\L$-orders to
evaluate such formulas instead of the standard order of the truth
values $0$ and~$1$.

In our setting, we formalize attribute implications as atomic formulas
in a language of algebras with $\L$-order: Consider
a set $Y = \{f_1,\ldots,f_n\}$ of attributes. Each attribute $f_i$ will
be considered as a nullary function symbol, i.e., as a symbol of an object
constant. In addition to that, we consider a binary function symbol
$\cdot$ (called a \emph{composition} which may be viewed as a symbol for
a fuzzy conjunction) and a nullary function symbol $\top$
(called an \emph{identity}). Therefore,
\begin{align}
  F = \{\cdotsym, f_1,\ldots,f_n, \top\}.
  \label{eqn:F}
\end{align}
Any inequality written in the language given by $F$ and $X = \emptyset$
is called a (\emph{general}) \emph{attribute implication.}

\begin{example}
  The role of the composition $\cdot$ is to express antecedents and
  consequent of attribute implications consisting of more than one attribute.
  For instance, \eqref{eqn:ai_pq} can be seen as inequality
  $p_1 \cdot (p_2 \cdot (\cdots p_m) \cdots) \ineq
  q_1 \cdot (q_2 \cdot (\cdots q_n) \cdots)$.
  In addition, $\top$ may be seen as the counterpart of the empty antecedents
  and consequents, e.g., $p \ineq \top$ and $\top \ineq q$ may represent
  $\{p\} \Rightarrow \emptyset$ and $\emptyset \Rightarrow \{q\}$.
\end{example}

In each algebra with $\L$-order which is considered a reasonable interpretation
of the generalized attribute implications, $\cdot$ and $\top$ shall satisfy some
basic properties. It is reasonable to assume that $\top$ is neutral with respect
to $\cdot$, $\top$ is the greatest element, $\cdot$ is associative
(to make parentheses in terms irrelevant) and commutative
(to make the order of $f_1,\ldots,f_n$ in terms irrelevant).
We therefore postulate the following laws:
\begin{align}
  t \cdot \top &\ineq t, \label{eqn:neu1} \\
  t &\ineq t \cdot \top, \label{eqn:neu2} \\
  t &\ineq \top, \label{eqn:int} \\
  r \cdot (s \cdot t) &\ineq (r \cdot s) \cdot t, \\
  (r \cdot s) \cdot t &\ineq r \cdot (s \cdot t), \\
  t \cdot s &\ineq s \cdot t, \label{eqn:com}
\end{align}
where $r,s,t \in T_F(\emptyset)$. Two remarks are in order: First, \eqref{eqn:int}
\emph{does not} ensure that an algebra $\st$ with $\L$-equality satisfying
\eqref{eqn:int} to degree $1$ has $\top^{\st}$ as the greatest element.
On the other hand, for each $f_i$, we have
$f^{\st}_i \ineq^{\st} \top^{\st} = 1$. Second, \eqref{eqn:com}
may be considered superfluous. It is the opinion of the author that $\cdot$
should be commutative but the logic can be developed in a more general setting
without \eqref{eqn:com} in much the same way as it is presented below.

\begin{definition}\label{def:sem}
  An algebra with $\L$-order of type \eqref{eqn:F} which satisfies
  inequalities \eqref{eqn:neu1}--\eqref{eqn:com}
  for $r,s,t \in T_F(\emptyset)$ to degree $1$ is called
  an \emph{$\L$-structure for general attribute implications}
  over attributes $Y = \{f_1,\ldots,f_n\}$.
\end{definition}

\begin{remark}
  Since we always consider the generalized attribute implications to be
  evaluated in $\L$-structures which are algebras with $\L$-orders satisfying 
  \eqref{eqn:neu1}--\eqref{eqn:com}, we may accept the usual rules
  of simplifying the inequalities. Namely, we disregard parentheses
  and the order of symbols in terms, and we may omit $\top$
  if it is a part of a compound term. In addition, we may omit the symbol
  of composition and write just $ts$ instead of $t \cdot s$.
  Therefore, \eqref{eqn:ai_pq} may be written as
  $p_1p_2\cdots p_m \ineq q_1q_2\cdots q_n$. Note that $\cdot$ is not
  idempotent and thus $p \ineq p$ and $p \ineq pp$ represent different
  general attribute implications.
\end{remark}

\begin{example}
  Let us show that particular $\L$-structure for general attribute implications
  can be derived directly from $\L$. Indeed, for a complete residuated lattice
  $\L = \langle L,\wedge,\vee,\otimes,\rightarrow,0,1\rangle$, we may consider
  a structure 
  \begin{align*}
    \st =
    \langle M,\ineq^{\st},\cdot^{\st},
    f^{\st}_1,\ldots,f^{\st}_n,\top^{\st}\rangle,
  \end{align*}
  where $M = L$,
  $f_i^{\st} \in L$ for each $i=1,\ldots,n$, $\top^{\st} = 1$, and
  \begin{align*}
    a \ineq^{\st} b &= a \rightarrow b, &
    a \cdot^{\st} b &= a \otimes b, 
  \end{align*}
  for all $a,b \in M$. It is easy to check that $\st$ is an algebra
  with $\L$-order and it satisfies each inequality
  \eqref{eqn:neu1}--\eqref{eqn:com} to degree $1$. First,
  $\st$ is indeed an algebra with $\L$-order:
  \eqref{eqn:iRef} is satisfied because
  $a \rightarrow b = b \rightarrow a = 1$ is true if{}f $a \leq b$
  and $b \leq a$ and thus if{}f $a = b$;
  \eqref{eqn:iTra} is satisfied because $(a \rightarrow b) \otimes
  (b \rightarrow c) \leq a \rightarrow c$ follows by the adjointness property;
  \eqref{eqn:iComF} is satisfied for $\cdot$ because
  $(a \rightarrow b) \otimes (c \rightarrow d) \leq
  (a \otimes c) \rightarrow (b \otimes d)$ holds in $\L$;
  the case of \eqref{eqn:iComF} and the nullary operations is trivial
  since $1 \leq f_i \rightarrow f_i$ and $1 \leq 1 \rightarrow 1$.
  In addition, each \eqref{eqn:neu1}--\eqref{eqn:com}
  is obviously satisfied to degree $1$ since $\langle L,\otimes,1\rangle$
  is a commutative monoid with $1$ being the greatest element in $L$.
  Thus, $\st$ represents an $\L$-structure for general attribute implications
  where $\cdot^{\st}$ is not idempotent in general.
  Note that an $\L$-structure for general attribute implications
  with idempotent $\cdot^{\st}$ may
  be obtained by putting $a \cdot^{\st} b = a \wedge b$ for all $a,b \in L$
  and leaving the rest as in the previous case. Again, using 
  $(a \rightarrow b) \otimes (c \rightarrow d) \leq
  (a \wedge c) \rightarrow (b \wedge d)$, it follows that
  the structure is indeed an $\L$-structure for general attribute implications.
\end{example}

The framework of the inequational logic gives us the notions of semantic
entailment and provability of general attribute implications:

\begin{definition}\label{def:ai_entail}
  Let $\Sigma$ by an $\L$-set of general attribute implications and let
  \begin{align}
    \Sigma^\mathrm{AI}(t, t') &=
    \left\{
      \begin{array}{@{\,}l@{\quad}l@{}}
        1, &\text{if } t \ineq t' \text{ is in the form of
          some formula in \eqref{eqn:neu1}--\eqref{eqn:com}}, \\
        0, &\text{otherwise.}
      \end{array}
    \right.
    \label{eqn:SigmaAI}
  \end{align}
  The degree $\val[\Sigma]{t \ineq t'}^{\mathrm{AI}}$ to which
  a general attribute implication $t \ineq t'$
  is \emph{semantically entailed} by $\Sigma$ is defined by
  \begin{align}
    \val[\Sigma]{t \ineq t'}^{\mathrm{AI}} &=
    \val[\Sigma \cup \Sigma^\mathrm{AI}]{t \ineq t'}
    \label{eqn:AI_sem}
  \end{align}
  and the degree $\prv[\Sigma]{t \ineq t'}^{\mathrm{AI}}$ to which
  $t \ineq t'$ is \emph{provable} by $\Sigma$ is defined by
  \begin{align}
    \prv[\Sigma]{t \ineq t'}^{\mathrm{AI}} &=
    \prv[\Sigma \cup \Sigma^\mathrm{AI}]{t \ineq t'},
    \label{eqn:AI_syn}
  \end{align}
  where $\Sigma \cup \Sigma^\mathrm{AI}$ denotes the union of
  $\L$-sets $\Sigma$ and $\Sigma^\mathrm{AI}$.
\end{definition}

Applying Theorem~\ref{th:compl}, we obtain the following completeness
of the logic of general attribute implications.

\begin{theorem}\label{th:compl_ais}
  Let $\Sigma$ by an $\L$-set of general attribute implications.
  Then, for any general attribute implication $t \ineq t'$, we have
  $\prv{t \ineq t'}^\mathrm{AI} = \val[\Sigma]{t \ineq t'}^\mathrm{AI}$.
\end{theorem}
\begin{proof}
  Consequence of \eqref{eqn:AI_sem}, \eqref{eqn:AI_syn},
  and Theorem~\ref{th:compl}.
\end{proof}

Let us conclude this section by remarks on the consequence of
Theorem~\ref{th:compl_ais} and properties of the proposed logic of general
attribute implications.

\begin{remark}
  Owing to the general notion of $\L$-structure
  for general attribute implications, the fact that
  $\val[\Sigma]{t \ineq t'}^\mathrm{AI} \geq a$ should be understood so that
  $t \ineq t'$ is true at least to degree $a$ under any possible
  interpretation of the composition $\cdot$ and the ordering $\ineq$ which
  makes all formulas true at least to the degrees prescribed by $\Sigma$. This
  is in contrast with the other approaches such as~\cite{BeVy:ADfDwG} where
  the analogues of the composition and ordering are given directly by the
  structure of degrees. In our setting, the structure of degrees just
  puts a constraint on the mutual relationship of $\cdot$ and $\ineq$.
  Namely, since $\st$ is supposed to be an algebra
  with $\L$-order, $\cdot^{\st}$ is compatible with $\ineq^{\st}$, i.e.,
  the condition~\eqref{eqn:iComF} with $\cdot$ in place of $f$. The condition
  is quite natural and generalizes the monotony property: If $t$ is less than
  or equal to $t'$ (under some evaluation) and
  $s$ is less than or equal to $s'$ (under the same evaluation),
  then $ts$ (i.e., the conjunction of $t$ and $s$) is less than or equal
  to $t's'$ (i.e., the conjunction of $t'$ and $s'$).
\end{remark}

\begin{remark}
  It is interesting to observe how the inference system simplifies
  in case of $F$ given by~\eqref{eqn:F} and $X = \emptyset$.
  First, $X = \emptyset$ means that~\eqref{eqn:Inv} is superfluous
  because it infers $\langle t \ineq t',a\rangle$ from
  $\langle t \ineq t',a\rangle$. In addition, in case of $f_1,\ldots,f_n$
  or $\top$, \eqref{eqn:Com} becomes a nullary rule which
  infers $\langle f_i \ineq f_i, 1\rangle$ or
  $\langle \top \ineq \top, 1\rangle$ from no input formulas.
  Since both are axioms to degree $1$, see~\eqref{eqn:Ax}, it makes sense
  to consider~\eqref{eqn:Com} only for the composition. That is, our
  deductive system for general attribute implications reduces to
  \begin{align}
    \mathrm{Tra}\!:\, &\cfrac{\langle t \ineq t', a\rangle,
      \langle t' \ineq t'', b\rangle}{
      \langle t \ineq t'', a \otimes b\rangle},
    &
    \mathrm{Com}\!:\, &\cfrac{
      \langle t \ineq t', a\rangle,
      \langle s \ineq s', b\rangle}{
      \langle ts \ineq t's',
      a \otimes b\rangle},
  \end{align}
  for all $t,t',t'',s,s' \in T_F(\emptyset)$ and $a,b \in L$.
  Observe that by a particular case of $\mathrm{Com}$ for
  $s = s'$ and $b = 1$, we get a derived deduction rule
  \begin{align}
    \mathrm{Aug}\!:\, &\cfrac{
      \langle t \ineq t', a\rangle}{
      \langle ts \ineq t's,
      a\rangle},
  \end{align}
  where $t,t',s \in T_F(\emptyset)$ and $a \in L$. Conversely, $\mathrm{Tra}$
  and $\mathrm{Aug}$ yield $\mathrm{Com}$. Indeed, applying $\mathrm{Aug}$
  twice, we get $\langle ts \ineq t's,a\rangle$ and
  $\langle st' \ineq s't',b\rangle$ from $\langle t \ineq t',a\rangle$ and
  $\langle s \ineq s',b\rangle$, respectively. Now, using the axiom of
  commutativity~\eqref{eqn:com} and $\mathrm{Tra}$,
  we infer $\langle t's \ineq t's',b\rangle$
  and thus $\langle ts \ineq t's', a \otimes b\rangle$ by $\mathrm{Tra}$.
  This shows that the deductive system can be reduced to $\mathrm{Tra}$
  and $\mathrm{Aug}$. This is an interesting observation because it means
  that the two deduction rules in our logic are in fact Pavelka-style extensions
  of the two main Armstrong deduction rules of transitivity and augmentation,
  see~\eqref{eqn:orig_Arms}. In addition, our system proves each
  $ts \ineq t$ to degree $1$ which generalizes the nullary
  Armstrong rule $\mathrm{Ax}$. Indeed, we infer
  $\langle st \ineq \top t,1\rangle$ from
  $\langle s \ineq \top,1\rangle$ by $\mathrm{Aug}$ and thus
  $\langle ts \ineq t,1\rangle$ is derivable by~\eqref{eqn:neu1}
  and \eqref{eqn:com} using $\mathrm{Tra}$. We can simplify the system
  even more by considering a single deduction rule which
  generalizes~\eqref{eqn:Cut_ai}. Namely, we may introduce
  \begin{align}
    \mathrm{Cut}\!:\, &\cfrac{
      \langle t \ineq t', a\rangle,
      \langle t's \ineq s', b\rangle}{
      \langle ts \ineq s',
      a \otimes b\rangle},
  \end{align}
  for all $t,t',s,s' \in T_F(\emptyset)$ and $a,b \in L$ with the possibility
  of $s$ being omitted. Clearly, $\mathrm{Tra}$ is then a particular case of
  $\mathrm{Cut}$ with $s$ omitted and $\mathrm{Aug}$ results by $\mathrm{Cut}$
  for $s' = t's$ and $b = 1$. Conversely, one can infer
  $\langle ts \ineq t's, a\rangle$ from $\langle t \ineq t', a\rangle$
  by $\mathrm{Aug}$ and then apply $\mathrm{Tra}$ with
  $\langle t's \ineq s', b\rangle$ to obtain the result of $\mathrm{Cut}$.
  Therefore, $\mathrm{Tra}$ and $\mathrm{Aug}$ can be replaced
  by $\mathrm{Cut}$. As a result, our logic has a Pavelka-style complete
  deductive system which results by attaching a non-trivial semantic part
  to the deduction rules of the ordinary Armstrong system in both the original
  version and the simplified version using $\mathrm{Cut}$.
\end{remark}

\subsubsection*{Conclusions}
We showed a Pavelka-style complete logic for reasoning with graded inequalities
using any complete residuated lattices as the structure of truth degrees. The
results generalize the previous results on completeness of fuzzy equational
logic by considering more general semantics given by algebras with fuzzy orders
and omitting the deduction rule of symmetry. In addition, we showed an
application of the general completeness result showing a way to generalize
the ordinary attribute implications in a graded setting with a general semantics
and Pavelka-style complete inference system which generalizes the well-known
Armstrong system of inference rules.

\subsubsection*{Acknowledgment}
Supported by grant no. \verb|P202/14-11585S| of the Czech Science Foundation.


\footnotesize
\bibliographystyle{amsplain}
\bibliography{fineql}

\end{document}